\documentclass{article}

\usepackage{amsmath}
\usepackage{amsfonts}
\usepackage{amssymb}
\usepackage{amsthm}
\usepackage{thmtools}
\usepackage{thm-restate}
\usepackage{comment}
\usepackage[utf8]{inputenc}
\usepackage[toc]{appendix}
\usepackage{algorithm}
\usepackage{algpseudocode}
\usepackage{wasysym,utfsym}
\usepackage{arydshln}
\usepackage[
backend=biber,
style=alphabetic,
maxbibnames=99,
sorting=nyt
]{biblatex}

\newtheorem{lemma}{Lemma}

\newtheorem{corollary}{Corollary}
\newtheorem{definition}{Definition}

\newtheorem{remark}{Remark}

\newcommand{\rs}{\text{RS}}
\newcommand{\bad}{\mathrm{Bad}}
\newcommand{\agree}{\mathrm{Agree}}
\newcommand{\rank}{\mathrm{rank}}
\newcommand{\mcaerr}{\operatorname{err}_{\mathrm{MCA}}}

\title{List-Decoding Counterexamples Yield Lower Bounds on Mutual Correlated Agreement Error}
\author{Yiwen Gao \and Hong Yang \and Yang Xu \and Haibin Kan\thanks{Yiwen Gao, Hong Yang, Yang Xu, and Haibin Kan are with the College of Computer Science and Artificial Intelligence, Fudan University, Shanghai 200433, China, and Shanghai Engineering Research Center of Blockchain, Shanghai 200433, China. Haibin Kan is also with Shanghai Institute for Mathematics and Interdisciplinary Sciences, Shanghai 200433, China. Emails: \texttt{ywgao21@m.fudan.edu.cn}, \texttt{hongyang1358@gmail.com}, \texttt{xuyyang@fudan.edu.cn}, \texttt{hbkan@fudan.edu.cn}.}}

\begin{document}
\maketitle

\begin{abstract}
Mutual correlated agreement captures whether a random linear combination of received words can create a new large agreement with a code, a property relevant to the soundness of batched proximity testing. We show constructively that list-decoding counterexamples yield lower bounds on the mutual correlated agreement error. Given an explicit counterexample to the $(p,L)$-list-decodability of a linear code over $\mathbb{F}_q$, we construct a related code $C'$ of the same length and dimension such that $\mcaerr(C',p)\ge\frac{1}{q}\left\lceil\frac{(L+1)q}{q+L}\right\rceil$, while decreasing its minimum distance by at most one. The construction also produces an explicit pair of words witnessing this error.

We further give a structure-preserving version for code families whose coordinates are indexed by a finite set $\Omega$, with each index determining a generator-matrix column through a map $v:\Omega\to\mathbb{F}_q^k$. The construction changes at most one coordinate index and ensures that the output code remains in the same indexed family. As applications, we instantiate this principle for algebraic-geometry (AG) evaluation codes and Reed--Solomon codes. For AG codes, if $G$ is the divisor defining the underlying Riemann--Roch space and $N$ is the number of rational places outside $\operatorname{supp}(G)$ available for evaluation, the resulting code remains over the same function field and Riemann--Roch space, with a modified set of evaluation places. Its mutual correlated agreement error is at least $\frac{1}{q}\left\lceil\frac{(L+1)N}{N+L\deg G}\right\rceil$. The Reed--Solomon conclusion follows as the Vandermonde-column specialization.
\end{abstract}

\section{Introduction}
Linear codes, defined as linear subspaces $C \subseteq \mathbb{F}_q^n$, play a vital role in modern digital communication and cryptography due to their elegant algebraic structure and efficient encoding mechanisms. In particular, they have emerged as a cornerstone in the design of succinct non-interactive arguments of knowledge (SNARKs). A critical component enabling these efficient proof systems is the \emph{proximity test}, a probabilistic algorithm that allows a verifier to efficiently determine whether a queried vector is close to a valid codeword by examining only a small number of its coordinates. In practice, it is often necessary to perform proximity tests on a large batch of vectors. To reduce verification overhead, a common approach is to randomly combine these vectors into a single vector and apply the proximity test exclusively to this linear combination. The soundness of such testing mechanisms inherently relies on the distance-preserving properties of the underlying codes.

Intuitively, the distance-preserving capability of a code guarantees that if a set of words is far from the code, a random linear combination of these words will also remain far from the code with high probability. This principle is central to the soundness analysis of IOPPs and code-based proof systems, and variants of this property have been studied in many works~\cite{RVW13,AHIV17,Ben+18,BGKS20}. To formally capture and quantify this behavior, several coding-theoretic notions have been introduced, including \emph{proximity gaps}, \emph{correlated agreement}, and \emph{mutual correlated agreement}; see, e.g.,~\cite{BCIKS23,ACF+25}. Mutual correlated agreement gives a more refined account of this phenomenon: it requires that a random linear combination not create a new large agreement domain with the code unless that domain is already explained by the correlated agreement structure of the original words. This stronger requirement is particularly useful when the soundness analysis must retain information about the individual words being combined, rather than merely show that their random combination is far from the code.

The study of (mutual) correlated agreement contains both positive results, which prove that the property holds in certain parameter regimes, and explicit lower-bound constructions, which show that the error probability cannot be made too small.

Positive results have been obtained for a variety of code families. For general linear codes, \cite{BGKS20,GKL24,Zei24} establish proximity-gap and related (mutual) correlated agreement up to the $1.5$ Johnson bound, while \cite{GCXK25} gives a general reduction from list decodability to mutual correlated agreement. For Reed--Solomon codes, \cite{BCIKS23} obtains proximity gaps up to the Johnson bound via the Guruswami--Sudan list-decoding algorithm; these guarantees are further improved by \cite{BCHKS25} and extended to mutual correlated agreement by \cite{Hab25}. Beyond these settings, \cite{GG25} establishes optimal proximity gaps for several structured and random code families, and \cite{YZ26} develops a syndrome-space approach for random linear codes.

On the negative side, for every constant $\tau$, \cite{BCHKS25} shows that $n^\tau$-bounded proximity-gap guarantees fail for certain Reed--Solomon code families, while \cite{CS25} shows that the error probability can equal $1$ for certain parameter choices. Beyond these results, \cite{BCHKS25} and \cite{CS25} establish connections between list decodability and correlated agreement for Reed--Solomon codes. \cite{DG25} derives a general lower bound on the correlated agreement error in terms of the probability that a random word is close to the code. For a comprehensive survey of these notions and their applications, we refer interested readers to~\cite{ABF26}.

\subsection{Our results}
Our main result shows that list-decoding counterexamples yield lower bounds on the mutual correlated agreement error. Recall that a code is $(p,L)$-list-decodable if every received word has at most $L$ codewords within relative Hamming distance $p$; equivalently, a counterexample to $(p,L)$-list-decodability consists of a received word with at least $L+1$ nearby codewords. Specifically, we demonstrate that any such counterexample, when the minimum distance is sufficiently larger than $p$, implies the existence of a related code $C'$ with $\mcaerr(C',p)\ge\frac{1}{q}\left\lceil\frac{(L+1)q}{q+L}\right\rceil$, where $q$ is the size of the underlying finite field. It has the same block length and dimension as the original code, and its relative minimum distance is preserved up to a possible loss of $1/n$ (equivalently, its absolute minimum distance decreases by at most one). For a word $c\in\mathbb{F}_q^n$, let $\operatorname{wt}(c)$ denote its Hamming weight and $\operatorname{supp}(c) \triangleq \{i\in[n]:c[i]\neq 0\}$ denote its support. This relationship is formally captured by our core theorem:

\begin{restatable}{theorem}{ThmLinearMCA}\label{thm:linear_mca}
    Let $C \subseteq \mathbb{F}_q^n$ be a linear code of dimension $k$ with generator matrix $G \in \mathbb{F}_q^{k\times n}$. If $n>k$, $C$ is not $(p,L)$-list-decodable, and $p<\Delta(C)-\frac{1}{n}$, then there exists a linear code $C' \subseteq \mathbb{F}_q^n$ of dimension $k$ such that
    \[
        \Delta(C')\ge \Delta(C)-\frac{1}{n},
    \]
    and
    \[
        \mcaerr(C',p)
        \ge
        \frac{1}{q}
        \left\lceil\frac{(L+1)q}{q+L}\right\rceil.
    \]
    Moreover, given a received word and $L+1$ nearby codewords witnessing that $C$ is not $(p,L)$-list-decodable, the construction explicitly produces both $C'$ and a pair of words witnessing this mutual correlated agreement error lower bound.

    Furthermore, if
    \[
        \bigcup_{\substack{c\in C\\ \operatorname{wt}(c)=d_{\min}(C)}}
        \operatorname{supp}(c) \neq [n],
    \]
    then the code $C'$ can be chosen to satisfy
    \[
        \Delta(C') \ge \Delta(C).
    \]
\end{restatable}

Theorem~\ref{thm:linear_mca} guarantees that the resulting code is linear, but it need not preserve any additional structure of the initial code. We next give a structure-preserving version for code families whose coordinates are indexed by a prescribed finite set $\Omega$, with each index determining a generator-matrix column through a map $v:\Omega\to\mathbb{F}_q^k$.

\begin{restatable}{theorem}{ThmStructuredMCA}\label{thm:structured_puncture_append}
Let $\Omega$ be a finite set with $|\Omega|=N>0$, and let
\(
    v:\Omega\to\mathbb{F}_q^k
\)
be a column map. Assume that there exists a number $\lambda$ such that, for every nonzero vector $u\in\mathbb{F}_q^k$,
\[
    \left|
        \{\omega\in\Omega:u\cdot v(\omega)^{\mathrm T}=0\}
    \right|
    \leq \lambda .
\]
Let $C\subseteq\mathbb{F}_q^n$ be a $k$-dimensional code with $n>k$ that has a full-row-rank generator matrix of the form
\[
    G=\bigl(v(\omega_1)\mid\cdots\mid v(\omega_n)\bigr)
\]
for distinct indices $\omega_1,\ldots,\omega_n\in\Omega$. If $C$ is not $(p,L)$-list-decodable and
\(
    p<\Delta(C)-\frac{1}{n},
\)
then there exists a $k$-dimensional code $C'\subseteq\mathbb{F}_q^n$ with a full-row-rank generator matrix whose columns have the form
\[
    G'=\bigl(v(\omega'_1)\mid\cdots\mid v(\omega'_n)\bigr)
\]
for distinct indices $\omega'_1,\ldots,\omega'_n\in\Omega$, such that
\(
    \left|
        \{\omega_1,\ldots,\omega_n\}
        \cap
        \{\omega'_1,\ldots,\omega'_n\}
    \right|
    \geq n-1,
\)
and
\(
    \Delta(C')\geq \Delta(C)-\frac{1}{n},
\)
and
\[
    \mcaerr(C',p)
    \geq
    \frac{1}{q}
    \left\lceil
        \frac{(L+1)N}{N+L\lambda}
    \right\rceil .
\]
Moreover, given a received word and $L+1$ nearby codewords witnessing that $C$ is not $(p,L)$-list-decodable, the construction explicitly produces $C'$ and a pair of words witnessing this mutual correlated agreement error lower bound.
\end{restatable}

Theorem~\ref{thm:structured_puncture_append} can be applied to algebraic-geometry (AG) evaluation codes and Reed--Solomon codes by taking the index set to be evaluation places or evaluation points and bounding their collisions. For AG codes, the changed coordinate is obtained by evaluating at a rational place, and a collision bound for functions in a Riemann--Roch space replaces the general linear-algebraic counting argument.

\begin{restatable}{corollary}{CorAGMCA}\label{cor:ag_mca}
Let $F/\mathbb{F}_q$ be an algebraic function field, let $G$ be a divisor of $F$, and let
\(
    D=P_1+\cdots+P_n
\)
be a sum of distinct rational places disjoint from $\operatorname{supp}(G)$. Suppose that $\deg G<n$, $n>\ell(G)$, the code $C_{\mathcal{L}}(D,G)$ is not $(p,L)$-list-decodable, and
\(
    p<\Delta\bigl(C_{\mathcal{L}}(D,G)\bigr)-\frac{1}{n}.
\)
Let
\[
    \mathcal{P}
    \triangleq
    \left\{
        P:\deg P=1,\ 
        P\notin\operatorname{supp}(G)
    \right\},
    \qquad
    N\triangleq|\mathcal{P}|.
\]
Then there exist an index $j\in[n]$ and a place $P\in\mathcal{P}\setminus\operatorname{supp}(D-P_j)$ such that, for
\(
    D'=D-P_j+P,
\)
the AG evaluation code $C_{\mathcal{L}}(D',G)$ satisfies
\[
    \mcaerr\bigl(C_{\mathcal{L}}(D',G),p\bigr)
    \geq
    \frac{1}{q}\left\lceil\frac{(L+1)N}{N+L\deg G}\right\rceil.
\]
Moreover, given a received word and $L+1$ nearby codewords witnessing that $C_{\mathcal{L}}(D,G)$ is not $(p,L)$-list-decodable, the construction explicitly produces $D'$ and a pair of words witnessing this mutual correlated agreement error lower bound. In particular, the resulting code remains an AG evaluation code over the same function field and Riemann--Roch space, has the same length and dimension as $C_{\mathcal{L}}(D,G)$, and satisfies
\[
    \Delta\bigl(C_{\mathcal{L}}(D',G)\bigr)
    \geq
    \Delta\bigl(C_{\mathcal{L}}(D,G)\bigr)-\frac{1}{n}.
\]
\end{restatable}

Thus, unlike the general linear-code reduction, the AG construction produces a code with large mutual correlated agreement error within the same evaluation-code family and over the same Riemann--Roch space, while changing at most one evaluation place. Specializing to the rational function field recovers Reed--Solomon codes, as stated formally in Corollary~\ref{cor:rs_structured_framework}: here $N=q$ and $\deg G=k-1$, yielding an error lower bound of
\(
    \frac{1}{q}
    \left\lceil\frac{(L+1)q}{q+L(k-1)}\right\rceil.
\)

\subsection{Technical overview}
We explain the main idea using Reed--Solomon codes as the guiding example. Let $\rs_{\mathbb{F}_q}(S,k)$ be a Reed--Solomon code evaluated on a domain $S\subseteq \mathbb{F}_q$, and suppose that it is not $(p,L)$-list-decodable. Then there exists a received word $y\in \mathbb{F}_q^S$ and $L+1$ distinct polynomials
\[
    P_1,\ldots,P_{L+1}\in \mathbb{F}_q[X],
    \qquad \deg(P_i)<k,
\]
such that each $P_i$ agrees with $y$ on at least a $(1-p)$ fraction of the points in $S$. Let
\[
    A_i \triangleq \{x\in S : P_i(x)=y(x)\}.
\]
Thus $|A_i|\geq (1-p)|S|$ for every $i$.

The basic idea is to turn these $L+1$ nearby polynomials into many bad combining points. To see this, first suppose that we are allowed to append one new evaluation point $x_0\notin S$. Define two words on $S\cup\{x_0\}$ by
\[
    \pi_0(x)=
    \begin{cases}
        y(x), & x\in S,\\
        0, & x=x_0,
    \end{cases}
    \qquad
    \pi_1(x)=
    \begin{cases}
        0, & x\in S,\\
        1, & x=x_0.
    \end{cases}
\]
The word $\pi_0$ carries the original received word, while $\pi_1$ is supported only on the newly appended point. Hence for every combining point $\alpha\in\mathbb{F}_q$, the combination $\pi_0+\alpha\pi_1$ is still equal to $y$ on the original domain $S$, and takes value $\alpha$ at $x_0$.

Now fix one of the nearby polynomials $P_i$. Since $P_i$ agrees with $y$ on $A_i$, the combination $\pi_0+\alpha\pi_1$ agrees with $P_i$ on $A_i$ for every choice of $\alpha$. If we choose
\[
    \alpha_i=P_i(x_0),
\]
then the same combination also agrees with $P_i$ at the new point $x_0$. Therefore $\pi_0+\alpha_i\pi_1$ has a large agreement with the Reed--Solomon code on $A_i\cup\{x_0\}$.

This is a newly generated agreement: the pair $(\pi_0,\pi_1)$ has the original correlated agreement on $A_i$, but it does not have correlated agreement on $A_i\cup\{x_0\}$, because $\pi_1$ is zero on $A_i$ but nonzero at $x_0$. Thus the scalar $\alpha_i=P_i(x_0)$ is a bad combining point.

Different polynomials may produce the same combining point, so we choose $x_0$ using a collision-aware averaging argument. Any pair $P_i,P_j$ can agree at only a limited number of field elements because $P_i-P_j$ is a nonzero polynomial of degree at most $k-1$. Averaging over the possible choices of $x_0\in\mathbb{F}_q$ therefore yields a point at which the values $P_1(x_0),\ldots,P_{L+1}(x_0)$ contain at least
\[
    \left\lceil
        \frac{(L+1)q}{q+L(k-1)}
    \right\rceil
\]
distinct values. Each distinct value gives a distinct bad combining point. Thus the construction remains effective even when collisions occur and does not require all $L+1$ evaluations to be distinct.

The append-only argument explains the source of the mutual correlated agreement error lower bound, but it increases the block length by one. To obtain a code of the same block length as the original one, we use the puncture-and-append procedure. We first choose the evaluation point $x_0$ by the averaging argument above. If $x_0\in S$, we puncture that point; otherwise we puncture any point of $S$. In either case, after puncturing one point we obtain a domain $S''$ that does not contain $x_0$, and the same list-decoding counterexample survives this puncturing step with only the expected rescaling of the relative distance parameter. We then append $x_0$ and apply the previous construction to
\[
    S' = S''\cup\{x_0\}.
\]
Thus the final Reed--Solomon code has the same block length as the original one and differs from it in at most one evaluation point.

The structured version follows the same pattern with field elements replaced by an index set $\Omega$ and evaluation columns replaced by a map $v:\Omega\to\mathbb{F}_q^k$. The averaging step only needs a collision bound for the set of indices on which two messages evaluate equally. For AG codes, this collision bound is supplied by the fact that a nonzero function in the relevant Riemann--Roch space has at most $\deg G$ rational zeros.

\subsection{Related work}

Ben-Sasson, Carmon, Haböck, Kopparty, and Saraf~\cite[Theorem~1.9]{BCHKS25} and Crites and Stewart~\cite[Theorem~2]{CS25} independently derive lower bounds on the correlated agreement error for Reed--Solomon codes from list-decoding counterexamples. More precisely, they show that sufficiently strong correlated agreement for $\rs_{\mathbb{F}_q}(S,k)$ implies list decodability of $\rs_{\mathbb{F}_q}(S,k+1)$ with an explicit list-size bound.

Their arguments use division by $X-\beta$, a Reed--Solomon-specific operation that produces the shift from degree $k$ to degree $k+1$. Our construction instead appends a coordinate, allowing it to apply to arbitrary linear codes and to lower-bound mutual correlated agreement error, but without this degree shift.

\section{Preliminaries}

We introduce the basic notation and definitions used throughout the paper. We first recall standard coding-theoretic notions, including linear codes, Hamming distance, minimum distance, list-decodability, and Reed--Solomon codes. We then formulate the correlated agreement notions that will be used to state failures of mutual correlated agreement.

\subsection{Coding-theoretic preliminaries}\label{subsec:coding_preliminaries}

All codes in this work are over a finite field $\mathbb{F}_q$. We use row-vector notation for messages, so a generator matrix maps a message $m\in\mathbb{F}_q^k$ to the codeword $mG$.

\begin{definition}[Linear code]\label{def:linear_code}
Let $\mathbb{F}_q$ be a finite field. A linear code $C$ of length $n$ and dimension $k$ over $\mathbb{F}_q$ is a $k$-dimensional linear subspace of $\mathbb{F}_q^n$. It can be represented by a $k \times n$ matrix $G$ of rank $k$ over $\mathbb{F}_q$, called the \emph{generator matrix}, such that $C = \{ m G \mid m \in \mathbb{F}_q^k \}$.
\end{definition}
Throughout this work, we always assume generator matrices have full row rank, so the corresponding code has dimension $k$.

Distance and agreement are always measured in relative Hamming distance.

\begin{definition}[Relative Hamming distance]\label{def:hamming}
Let $\pi_1,\pi_2\in \mathbb{F}_q^n$ be two words, where $n\in\mathbb{N}^+$. Define the relative Hamming distance between $\pi_1$ and $\pi_2$ to be\[ \Delta(\pi_1,\pi_2) \triangleq \frac{|\{j\in \{1,\ldots,n\}: \pi_1[j]\neq \pi_2[j]\}|}{n}.
\]
We denote by $\agree(\pi_1, \pi_2) \triangleq \{j\in \{1,\ldots,n\}: \pi_1[j] = \pi_2[j]\}$ the set of agreement points between $\pi_1$ and $\pi_2$. For a word $\pi\in\mathbb{F}_q^n$, we write $\mathrm{wt}(\pi)\triangleq |\{j\in\{1,\ldots,n\}:\pi[j]\neq 0\}|$ for its Hamming weight.

Let $C\subseteq \mathbb{F}_q^n$ be a set of codewords with length $n$ and $\pi\in \mathbb{F}_q^n$ be a word. Define the relative Hamming distance between $\pi$ and $C$ to be\[ \Delta(\pi,C) \triangleq \min_{c\in C} \Delta(\pi,c).
\]
\end{definition}

\begin{definition}[Minimum distance of a linear code]\label{def:minimum_distance}
Let $C \subseteq \mathbb{F}_q^n$ be a linear code. The relative minimum distance of $C$ is defined as
\[
    \Delta(C) \triangleq \min_{\substack{c,c'\in C\\ c\neq c'}} \Delta(c,c').
\]
Equivalently, since $C$ is linear,
\[
    \Delta(C)=\min_{0\neq c\in C}\Delta(c,0).
\]
We also write
\[
    d_{\min}(C)\triangleq n\Delta(C)
\]
for the absolute minimum distance.
\end{definition}

List-decodability bounds the number of codewords that can lie within a prescribed radius of any received word.

\begin{definition}[List decoding]\label{def:list_decoding}
    Let $C\subseteq \mathbb{F}_q^n$ be a code. We say that $C$ is $(p, L)$-list-decodable if for every word $\pi\in \mathbb{F}_q^n$, we have\[ \left|\{c\in C\mid \Delta(c,\pi)\le p\}\right|\le L.
    \]
\end{definition}

We will also consider Reed--Solomon codes, where coordinates are indexed by distinct evaluation points rather than by the integers $\{1,\dots,n\}$.

\begin{definition}[Reed--Solomon code]\label{def:rs_code}
Let $\mathbb{F}_q$ be a finite field and $S=\{s_1,\ldots,s_n\}\subseteq\mathbb{F}_q$ be a set of distinct evaluation points. The Reed--Solomon code of length $n$ and dimension $k$ over $\mathbb{F}_q$, evaluated on $S$, is defined as the evaluations of all polynomials in $\mathbb{F}_q[X]$ of degree strictly less than $k$ on $S$:
\[
\rs_{\mathbb{F}_q}(S, k) \triangleq \{ (P(s_1), \ldots, P(s_n)) \mid P \in \mathbb{F}_q[X], \deg(P) < k \}.
\]
\end{definition}
Denote by $\mathbb{F}_q^S$ the set of all functions from $S$ to $\mathbb{F}_q$. The Reed--Solomon code $\rs_{\mathbb{F}_q}(S, k)$ is a $k$-dimensional linear subspace of $\mathbb{F}_q^S$, and for Reed--Solomon codes indexed by a set $S$, the preceding distance, agreement, minimum-distance, and list-decoding definitions are interpreted with the coordinate set $S$ in place of $\{1,\dots,n\}$. Its relative minimum distance is
\[
    \Delta(\rs_{\mathbb{F}_q}(S,k)) = 1-\frac{k-1}{|S|},
\]
so Reed--Solomon codes are maximum-distance separable (MDS) codes. The distinction of the evaluation points becomes important later when we puncture and append coordinates while preserving the Reed--Solomon structure.

\subsection{Mutual Correlated Agreement}

In related studies of proximity gaps for linear codes, the central object is often the collective behavior of a sequence of words rather than each word in isolation. This joint behavior is captured by \emph{(mutual) correlated agreement}, which has been formalized in recent work \cite{Ben+18, ACF+25, Arn+24}.

We first recall the standard definition of correlated agreement, which captures the scenario where a sequence of words shares a common, dense set of indices on which they agree with codewords.

\begin{definition}[Correlated Agreement]
Let $C \subseteq \mathbb{F}_q^n$ be a linear code, $0 \le \delta \le 1$, and let $\{\pi_0, \dots, \pi_l\} \subseteq \mathbb{F}_q^n$ be a sequence of words. We say $\pi_0, \dots, \pi_l$ have correlated agreement with density $\ge 1-\delta$ if there exists a set of indices $D \subseteq \{1, \dots, n\}$ and codewords $c_0, \dots, c_l \in C$ satisfying:
\begin{itemize}
    \item \textbf{Density:} $|D| \ge (1-\delta)n$.
    \item \textbf{Agreement:} For all $i \in \{0, \dots, l\}$, $\pi_i$ and $c_i$ agree on $D$.
\end{itemize}
\end{definition}

To formally characterize mutual correlated agreement, we must consider the collection of all \emph{maximal} $\delta$-correlated agreement sets of indices for a given sequence of words, denoted as $\mathcal{A}_{\delta, \{\pi_0, \dots, \pi_l\}, C}$. Such a set is maximal if no proper superset also satisfies the correlated agreement property.

Mutual correlated agreement dictates that a random linear combination of words will not spuriously generate a new, large agreement set of indices with the code unless that set inherently stems from one of the pre-existing maximal sets. This property is defined by strictly bounding the density of ``bad'' combining points.

\begin{definition}[Bad Combining Points]\label{def:bad_combining}
Define the set of bad combining points $\bad_{C,\delta}(\pi_0, \pi_1) \subseteq \mathbb{F}_q$ as:
\begin{equation}
\begin{aligned}
    \bad_{C,\delta}(\pi_0, \pi_1) \triangleq 
    \quad \left\{ \alpha\in \mathbb{F}_q \;\middle|\; 
    \begin{aligned}
        &\exists c \in C, \text{ such that } \Delta(\pi_0 + \alpha \pi_1, c) \le \delta \\
        &\text{and } \agree(\pi_0 + \alpha \pi_1, c) \notin \mathcal{A}_{\delta, \{\pi_0, \pi_1\}, C}
    \end{aligned}
    \right\}.
\end{aligned}
\end{equation}
\end{definition}

\begin{definition}[Mutual Correlated Agreement]\label{def:mca_error}
For a linear code $C\subseteq\mathbb{F}_q^n$ and $0\leq\delta\leq 1$, define its mutual correlated agreement error at radius $\delta$ by
\[
    \mcaerr(C,\delta)
    \triangleq
    \max_{\pi_0,\pi_1\in\mathbb{F}_q^n}
    \Pr_{\alpha\in\mathbb{F}_q}
    \bigl(\alpha\in\bad_{C,\delta}(\pi_0,\pi_1)\bigr).
\]
Let $\epsilon>0$. We say that a linear code $C\subseteq\mathbb{F}_q^n$ satisfies the $(\delta,\epsilon)$-mutual correlated agreement property if
\[
    \mcaerr(C,\delta)<\epsilon.
\]
\end{definition}

The definitions above extend directly to linear combinations of more than two words. Since this work focuses on combinations of two words, we use the two-word version throughout.

\section{From List-Decoding Counterexamples to Mutual Correlated Agreement Error Lower Bounds}\label{sec:linear_code}
In this section, we establish a formal connection between the list-decodability of general linear codes and their mutual correlated agreement error. Specifically, the following theorem shows that a list-decoding counterexample below the minimum-distance threshold yields a lower bound on the mutual correlated agreement error of a related linear code.
\ThmLinearMCA*
We prove this result constructively by showing how a list-decoding counterexample for a base linear code yields a related code with large mutual correlated agreement error. The proof proceeds in three parts:
\begin{itemize}
    \item In Section~\ref{subsec:out_of_domain}, we introduce an ``out-of-domain extension'' that appends a carefully chosen coordinate to the base code, producing an extended code and a pair of words with a large set of bad combining points.
    \item In Section~\ref{subsec:puncture_append}, we first puncture a rank-preserving coordinate and then apply the out-of-domain extension. This ``puncture one and append one'' construction preserves the original block length and dimension, produces a large set of bad combining points, and loses at most $1/n$ in relative minimum distance.
    \item In Section~\ref{subsec:distance_restoration}, we analyze this possible distance loss. We prove the no-loss case in Theorem~\ref{thm:linear_mca} and, in the complementary full-support case, identify conditions under which the appended coordinate restores the original minimum distance while preserving the bad-combining-point construction.
\end{itemize}

\subsection{Out of domain extension}\label{subsec:out_of_domain}

In this subsection, we introduce the ``out-of-domain extension'' technique. We prove that any linear code of block length $n$ failing to be $(p,L)$-list-decodable can be extended to a code of block length $n+1$ which, alongside a carefully chosen pair of words, yields a large set of bad combining points.

The intuition behind this construction relies on a list-decoding counterexample. If a code is not list-decodable, there exists a dense cluster of codewords around some received word $y$. We append a new coordinate to the code's generator matrix and construct a pair of words $\pi_0$ and $\pi_1$ such that $\pi_0$ preserves the original agreements with the clustered codewords, while $\pi_1$ isolates the newly added coordinate. For each of the clustered codewords, we can find a combining point $\alpha$ where the linear combination $\pi_0 + \alpha \pi_1$ aligns with the extended codeword at the appended position. Collisions among these combining points may occur, so we first record a simple linear-algebraic collision bound.

\begin{lemma}\label{lem:linear_collision_aware_out_of_domain}
    Let $m_1,\dots,m_{L+1}\in\mathbb{F}_q^k$ be distinct vectors. Then there exists a vector $a\in\mathbb{F}_q^k$ such that
    \[
        \left|\{m_i\cdot a^{\mathrm T}:1\le i\le L+1\}\right|
        \ge
        \left\lceil
        \frac{(L+1)q}{q+L}
        \right\rceil .
    \]
\end{lemma}

\begin{proof}
    For each $1\le i<j\le L+1$, define $V_{ij}\triangleq \{a\in\mathbb{F}_q^k:(m_i-m_j)\cdot a^{\mathrm T}=0\}$. Since $m_i\neq m_j$, the vector $m_i-m_j$ is nonzero. Hence $V_{ij}$ is a hyperplane in $\mathbb{F}_q^k$ and has size $q^{k-1}$.

    For $a\in\mathbb{F}_q^k$, let $C(a)$ be the number of colliding pairs among the values $m_1\cdot a^{\mathrm T},\dots,m_{L+1}\cdot a^{\mathrm T}$, namely
    \[
        C(a)\triangleq
        \left|\{(i,j):1\le i<j\le L+1,\ m_i\cdot a^{\mathrm T}=m_j\cdot a^{\mathrm T}\}\right|.
    \]
    For a fixed pair $(i,j)$, the collision condition is exactly $a\in V_{ij}$. Therefore
    \[
        \sum_{a\in\mathbb{F}_q^k} C(a)
        =
        \sum_{1\le i<j\le L+1}|V_{ij}|
        =
        \binom{L+1}{2}q^{k-1}.
    \]
    By the pigeonhole principle, there exists $a\in\mathbb{F}_q^k$ such that $C(a)\le \binom{L+1}{2}/q$.
    
    Fix such an $a$. Suppose that the $L+1$ values $m_1\cdot a^{\mathrm T},\ldots,m_{L+1}\cdot a^{\mathrm T}$ take exactly $r$ distinct values, with multiplicities $b_1,\ldots,b_r$. Then $\sum_{\ell=1}^r b_\ell=L+1$ and $C(a)=\sum_{\ell=1}^r\binom{b_\ell}{2}$, so $\sum_{\ell=1}^r b_\ell^2=(L+1)+2C(a)$. By Cauchy's inequality,
    \[
        (L+1)^2
        =
        \left(\sum_{\ell=1}^r b_\ell\right)^2
        \le
        r\sum_{\ell=1}^r b_\ell^2
        =
        r\bigl((L+1)+2C(a)\bigr).
    \]
    Using the choice of $a$, we obtain
    \[
        r
        \ge
        \frac{(L+1)^2}{(L+1)+2C(a)}
        \ge
        \frac{(L+1)^2}{(L+1)+(L+1)L/q}
        =
        \frac{(L+1)q}{q+L}.
    \]
    Since $r$ is an integer, $r\ge \left\lceil \frac{(L+1)q}{q+L}\right\rceil$.
\end{proof}

\begin{lemma}\label{lem:out-of-domain}
    Let $C \subseteq \mathbb{F}_q^n$ be a linear code with generator matrix $G \in \mathbb{F}_q^{k\times n}$. If $C$ is not $(p,L)$-list-decodable and $p<\Delta(C)$, then there exists a vector $a \in \mathbb{F}_q^k$ and a pair of words $\pi_0, \pi_1 \in \mathbb{F}_q^{n+1}$ such that for the extended linear code $C' \subseteq \mathbb{F}_q^{n+1}$ generated by the matrix $\begin{pmatrix} G & a^{\mathrm{T}} \end{pmatrix}$, we have
    \begin{align*}
        \left|\bad_{C',p\frac{n}{n+1}}(\pi_0,\pi_1)\right| \ge \left\lceil
        \frac{(L+1)q}{q+L}
        \right\rceil.
    \end{align*}
\end{lemma}

\begin{proof}
    We prove this by an explicit construction.

    Since $C$ is not $(p, L)$-list-decodable, there exists a received vector $y \in \mathbb{F}_q^n$ and $L+1$ distinct codewords $c_1, \dots, c_{L+1} \in C$ such that for each $i \in \{1, \dots, L+1\}$, the relative Hamming distance satisfies:
    \[
    \Delta(c_i, y) \le p.
    \]
    Let $A_i \triangleq \mathrm{agree}(y, c_i)$ denote the set of coordinates in $\{1, \dots, n\}$ where $y$ and $c_i$ agree. By the definition of relative Hamming distance, we know that $|A_i| \ge (1-p)n$ for all $i \in \{1, \dots, L+1\}$.

    Each codeword $c_i$ is generated by a unique message vector $m_i \in \mathbb{F}_q^k$, i.e., $c_i=m_iG$. Since the codewords $c_i$ are distinct and $G$ has full row rank, the messages $m_i$ are distinct.

    By Lemma~\ref{lem:linear_collision_aware_out_of_domain}, there exists $a\in\mathbb{F}_q^k$ such that the set
    \[
        T\triangleq \{m_i\cdot a^{\mathrm T}:1\le i\le L+1\}
    \]
    has size at least $\left\lceil \frac{(L+1)q}{q+L}\right\rceil$. Fix such an $a$.
    
    Define $\pi_0, \pi_1 \in \mathbb{F}_q^{n+1}$ by:
    \begin{align*}
        \pi_0 = (y, 0) \quad \text{and} \quad \pi_1 = (\overbrace{0,\dots,0}^{n}, 1).
    \end{align*}
    Let $C'$ be the linear code generated by $\begin{pmatrix}G&a^{\mathrm T}\end{pmatrix}$. Notice that $\pi_1$ is identically zero on the first $n$ coordinates (i.e., $\pi_1|_{A_i} = 0$), but its $(n+1)$-th coordinate is $\pi_1[n+1] = 1 \neq 0$. We claim that $\pi_0$ and $\pi_1$ do not have correlated agreement on $A_i \cup \{n+1\}$ for all $i \in \{1, \dots, L+1\}$. Indeed, if $\pi_1$ agreed with some codeword $(mG,m\cdot a^{\mathrm{T}})\in C'$ on $A_i\cup\{n+1\}$, then $(mG)|_{A_i}=0$ and $m\cdot a^{\mathrm{T}}=1$. In particular, $m\neq 0$. Since $G$ has full row rank, $mG$ is a nonzero codeword of $C$. However,
    \[
        \mathrm{wt}(mG)\le n-|A_i|\le pn<n\Delta(C),
    \]
    contradicting the definition of the relative minimum distance of $C$.

    We then show that every value in $T$ is a bad combining point.

    Fix any $\alpha\in T$, and choose an index $i$ such that $\alpha=m_i\cdot a^{\mathrm T}$. Let $c'_i=(c_i,m_i\cdot a^{\mathrm T})\in C'$. On the coordinates in $A_i$, we have $(\pi_0+\alpha\pi_1)|_{A_i}=c_i|_{A_i}=c'_i|_{A_i}$. On the appended coordinate,
    \[
        \pi_0[n+1]+\alpha\pi_1[n+1]=\alpha=m_i\cdot a^{\mathrm T}=c'_i[n+1].
    \]
    Hence $\pi_0+\alpha\pi_1$ agrees with $c'_i$ on $A_i\cup\{n+1\}$. The size of this agreement region is $|A_i \cup \{n+1\}| = |A_i| + 1 \ge (1-p)n + 1 = \big(1 - p \frac{n}{n+1}\big)(n+1)$. Since $\pi_0$ and $\pi_1$ do not have correlated agreement on $A_i\cup\{n+1\}$, this shows that $\alpha\in \bad_{C',p\frac{n}{n+1}}(\pi_0,\pi_1)$. Therefore,
    \begin{align*}
        |T|
        \le
        \left|\bad_{C',p\frac{n}{n+1}}(\pi_0, \pi_1)\right|.
    \end{align*}
    The lower bound on $|T|$ completes the proof.
\end{proof}
\begin{remark}
    The same construction also extends to the $l$-wise linear-combination setting. Namely, for $l\ge 2$, consider
    \[
        \pi_0+\alpha_1\pi_1+\cdots+\alpha_{l-1}\pi_{l-1},
        \qquad
        (\alpha_1,\dots,\alpha_{l-1})\in\mathbb{F}_q^{l-1},
    \]
    and set
    \[
        \pi_0=(y,0), \qquad
        \pi_1=\cdots=\pi_{l-1}=(\overbrace{0,\dots,0}^{n},1).
    \]
    For each distinct value $\beta\in T$, every tuple satisfying
    \[
        \alpha_1+\cdots+\alpha_{l-1}=\beta
    \]
    yields a newly generated agreement on $A_i\cup\{n+1\}$ for any $i$ with $\beta=m_i\cdot a^{\mathrm T}$. The corresponding equations for distinct $\beta$ define disjoint affine hyperplanes in $\mathbb{F}_q^{l-1}$, each of size $q^{l-2}$. Thus the density of bad combining tuples is at least
    \[
        \frac{|T|q^{l-2}}{q^{l-1}}
        \ge
        \frac{1}{q}
        \left\lceil
        \frac{(L+1)q}{q+L}
        \right\rceil .
    \]
\end{remark}

\subsection{Puncture one and append one}\label{subsec:puncture_append}

While the ``out-of-domain extension'' presented in the previous subsection generates a large set of bad combining points, it increases the block length of the base linear code. In this subsection, we address this limitation by introducing a ``puncture one and append one'' approach. By first deleting a coordinate from the initial linear code (puncturing) and subsequently appending a new coordinate using the out-of-domain extension, we can construct a code and a pair of words whose corresponding set of bad combining points satisfies the same lower bound as in Lemma~\ref{lem:out-of-domain}. Furthermore, this ``puncture one and append one'' method ensures that the resulting code preserves the same block length and relative distance bound as the original list-decoding counterexample.

\begin{lemma}\label{lem:puncture_and_append}
    Let $C \subseteq \mathbb{F}_q^n$ be a linear code of dimension $k$ with generator matrix $G \in \mathbb{F}_q^{k\times n}$. If $n>k$, $C$ is not $(p,L)$-list-decodable, and $p<\Delta(C)-\frac{1}{n}$, then there exists a linear code $C'\subseteq \mathbb{F}_q^n$ with generator matrix $G'\in\mathbb{F}_q^{k\times n}$ and a pair of words $\pi_0, \pi_1 \in \mathbb{F}_q^{n}$ such that
    \begin{align*}
        \left|\bad_{C', p}(\pi_0,\pi_1)\right|
        \ge
        \left\lceil
        \frac{(L+1)q}{q+L}
        \right\rceil .
    \end{align*}
    Moreover, the constructed code satisfies
    \[
        \Delta(C')\ge \Delta(C)-\frac{1}{n}.
    \]
\end{lemma}

To prove Lemma~\ref{lem:puncture_and_append}, we first introduce the following puncturing lemma. This lemma establishes that we can puncture a coordinate from the initial code without reducing its dimension while maintaining its non-list-decodability.

\begin{lemma}[Rank-Preserving Puncturing Lemma]\label{lem:puncture}
    Let $C \subseteq \mathbb{F}_q^n$ be a linear code of length $n$ and dimension $k$, with $n > k$. If $C$ is not $(p, L)$-list-decodable, then there exists a coordinate index $i^* \in [n]$ such that the punctured code $C' \subseteq \mathbb{F}_q^{n-1}$ obtained by deleting the $i^*$-th coordinate has dimension $k$ (i.e., its generator matrix maintains full row rank) and is not $\left(\frac{pn}{n-1}, L\right)$-list-decodable. In addition, its absolute minimum distance satisfies
    \[
        d_{\min}(C')\in\{d_{\min}(C),\,d_{\min}(C)-1\}.
    \]
\end{lemma}

\begin{proof}
    Let $G$ be a $k \times n$ generator matrix of $C$ with rank $k$. Since $n>k$, the set of columns of $G$ is linearly dependent. Therefore, there exists a column $G_{i^*}$ that is a linear combination of the remaining columns. Puncturing at coordinate $i^*$ deletes this column, and the remaining columns still span a $k$-dimensional row space. Hence the punctured code $C'$ has dimension $k$.

    Let $D=d_{\min}(C)$. Puncturing a coordinate decreases the weight of each nonzero codeword by either $0$ or $1$. Since the puncturing above preserves dimension, nonzero codewords remain nonzero after puncturing. Hence every nonzero codeword of $C'$ has weight at least $D-1$, while any minimum-weight codeword of $C$ punctures to a nonzero codeword of weight either $D$ or $D-1$. Therefore
    \[
        d_{\min}(C')\in\{D,D-1\}.
    \]

    For the list-decodability part, since $C$ is not $(p, L)$-list-decodable, there exist a received word $y \in \mathbb{F}_q^n$ and distinct codewords $c_1,\dots,c_{L+1} \in C$ such that $d(c_i, y) \le pn$ for all $i$, where $d(\cdot, \cdot)$ denotes the absolute Hamming distance. Let $y'$ and $c_i'$ be obtained by deleting coordinate $i^*$. Then $d(c_i', y') \le d(c_i, y) \le pn$, so
    \[
        \Delta(c_i', y') \le \frac{pn}{n-1}.
    \]
    Thus $y'$ has at least $L+1$ codewords of $C'$ within relative distance $\frac{pn}{n-1}$, and $C'$ is not $\left(\frac{pn}{n-1}, L\right)$-list-decodable.
\end{proof}

\begin{proof}[Proof of Lemma~\ref{lem:puncture_and_append}]
    By Lemma~\ref{lem:puncture}, there exists a punctured code $\widetilde{C} \subseteq \mathbb{F}_q^{n-1}$ that is not $\left(p\frac{n}{n-1}, L\right)$-list-decodable with full row rank generator matrix $\widetilde{G}\in \mathbb{F}_q^{k\times(n-1)}$, i.e., $\rank(\widetilde{G})=k$.

    Puncturing can decrease the absolute weight of a nonzero codeword by at most one. Since the puncturing above preserves dimension, every nonzero codeword of $\widetilde{C}$ comes from a nonzero codeword of $C$. Hence
    \[
        \Delta(\widetilde{C})\ge \frac{n\Delta(C)-1}{n-1}.
    \]
    The assumption $p<\Delta(C)-\frac{1}{n}$ therefore implies
    \[
        p' \triangleq p\frac{n}{n-1} < \Delta(\widetilde{C}).
    \]
    We can therefore apply Lemma~\ref{lem:out-of-domain} to $\widetilde{C}$ with parameters $p'$ and list size $L$ to obtain an extended code $C' \subseteq \mathbb{F}_q^{n}$ and words $\pi_0,\pi_1 \in \mathbb{F}_q^n$ such that
    \[
        \left|\bad_{C',\,p'\frac{n-1}{n}}(\pi_0,\pi_1)\right|
        \ge
        \left\lceil
        \frac{(L+1)q}{q+L}
        \right\rceil .
    \]
    Since appending one coordinate cannot decrease the absolute weight of any codeword, the same lower bound on absolute minimum distance gives
    \[
        d_{\min}(C')\ge n\Delta(C)-1,
        \qquad\text{and hence}\qquad
        \Delta(C')\ge \Delta(C)-\frac{1}{n}.
    \]
    Since $p'\frac{n-1}{n} = p$, the claim follows.
\end{proof}

With Lemma~\ref{lem:puncture_and_append} established, the first part of our main result Theorem~\ref{thm:linear_mca} follows immediately.

\subsection{Minimum-distance behavior under puncture and append}\label{subsec:distance_restoration}
This subsection proves the remaining assertion of Theorem~\ref{thm:linear_mca}: when the supports of the minimum-weight codewords of $C$ do not cover all coordinates, the constructed code $C'$ can be chosen with no loss in relative minimum distance. We also discuss the complementary full-support case and give conditions under which the appended coordinate restores the original minimum distance.

The construction in the proof of Lemma~\ref{lem:puncture_and_append} may lose at most $1/n$ in relative minimum distance, namely,
\[
    \Delta(C')\ge \Delta(C)-\frac{1}{n},
\]
which follows from the fact that puncturing can reduce the absolute weight of a nonzero codeword by at most one, while appending a coordinate cannot reduce any Hamming weight. We now discuss when this loss actually occurs and when the appended coordinate can restore the original minimum distance.

Let
\[
    \mathcal{W}_{\min}
    \triangleq
    \{c\in C:\operatorname{wt}(c)=d_{\min}(C)\}
\]
be the set of all minimum-weight codewords of $C$. If their supports do not cover all coordinates, namely
\[
    \bigcup_{c\in\mathcal{W}_{\min}}\operatorname{supp}(c)\neq [n],
\]
then there exists a coordinate $i\in[n]$ that is zero on every minimum-weight codeword. Puncturing this coordinate does not decrease the minimum distance: the minimum-weight codewords keep weight $d_{\min}(C)$, while every codeword of weight at least $d_{\min}(C)+1$ can lose at most one nonzero coordinate and therefore still has weight at least $d_{\min}(C)$ after puncturing. Moreover, deleting such a coordinate preserves rank. Indeed, if the rank dropped after deleting coordinate $i$, then there would exist a nonzero message $m\in\mathbb{F}_q^k$ such that $mG$ is zero on every coordinate except possibly $i$. Since $G$ has full row rank, $mG$ is a nonzero codeword of $C$, so it must be nonzero at coordinate $i$ and hence has weight $1$. Therefore $mG$ is a minimum-weight codeword whose support contains $i$, contradicting the choice of $i$.

To summarize, under the additional condition that the minimum-weight codewords do not support all coordinates, the puncturing step in Lemma~\ref{lem:puncture_and_append} can be chosen so that it does not decrease the minimum distance. Consequently, in this case Theorem~\ref{thm:linear_mca} can be strengthened to produce a code $C'$ with
\[
    \Delta(C')\ge \Delta(C),
\]
and hence the constructed code has relative minimum distance at least that of the initial code. This completes the proof of Theorem~\ref{thm:linear_mca}. For completeness, we next discuss the complementary full-support case.

When the minimum-weight codewords cover all coordinates, namely
\[
    \bigcup_{c\in\mathcal{W}_{\min}}\operatorname{supp}(c)=[n],
\]
every coordinate lies in the support of some minimum-weight codeword. In this case, any rank-preserving puncturing step necessarily decreases the absolute minimum distance by one.

Thus, in the full-support case, the only way to recover the original distance is to use the appended coordinate. For an appended column $a\in\mathbb{F}_q^k$, let
\[
    G'=(\widetilde{G}\mid a),
    \qquad
    C'=\{mG':m\in\mathbb{F}_q^k\}.
\]
The distance-critical messages are precisely the nonzero vectors $m\in\mathbb{F}_q^k$ whose punctured codewords have weight $D-1$. The appended coordinate restores the original absolute minimum distance if
\begin{equation}\label{eq:distance_restoration_condition}
    m\cdot a\neq 0
    \qquad
    \text{for every }m\in\mathbb{F}_q^k\setminus\{0\}
    \text{ such that }\operatorname{wt}(m\widetilde{G})=D-1.
\end{equation}
Equivalently, $a$ must avoid the union of the hyperplanes
\[
    \bigcup_{\substack{m\in\mathbb{F}_q^k\setminus\{0\}\\
    \operatorname{wt}(m\widetilde{G})=D-1}}
    \{x\in\mathbb{F}_q^k:m\cdot x=0\}.
\]
In addition, the out-of-domain construction requires $a$ to separate the $L+1$ messages $m_1,\dots,m_{L+1}$ from the list-decoding counterexample:
\begin{equation}\label{eq:message_separation_condition}
    (m_r-m_s)\cdot a\neq 0
    \qquad
    \text{for all }1\le r<s\le L+1.
\end{equation}
Thus, in the full-support case, the appended column must satisfy both conditions~\eqref{eq:distance_restoration_condition} and~\eqref{eq:message_separation_condition}: the first restores the original distance, and the second preserves the bad-combining-point construction.

We conclude with two simple examples illustrating the two cases. Let
\[
    C_0=\{(a,a,b,a+b):a,b\in\mathbb{F}_q\}\subseteq\mathbb{F}_q^4.
\]
Its minimum-weight codewords are the nonzero words of the form $(0,0,b,b)$, whose support is $\{3,4\}$; hence the minimum-weight codewords do not cover all coordinates. In our construction, puncturing either the first or the second coordinate preserves the minimum distance, and the subsequent extension cannot decrease it. Thus, the resulting code has the same minimum distance as $C_0$. In contrast, let $S\subseteq\mathbb{F}_q$ have size $n$ and let $C_1=\rs_{\mathbb{F}_q}(S,k)$, where $1\leq k<n$. For every $x\in S$, choose a subset $T\subseteq S\setminus\{x\}$ of size $k-1$. The polynomial
\[
    \prod_{t\in T}(X-t)
\]
defines a minimum-weight codeword of $C_1$ that is nonzero at $x$. Thus, the minimum-weight codewords of $C_1$ cover all evaluation coordinates.
\section{Structure-Preserving Puncture and Append}
The construction in Section~\ref{sec:linear_code} allows the appended column to be an arbitrary vector in $\mathbb{F}_q^k$. This is sufficient for general linear codes, but it may destroy additional structure: for example, an appended column for a Reed--Solomon generator matrix should be a Vandermonde column, and an appended column for an AG evaluation code should come from evaluation at a rational place. In this section, we isolate a simple coordinate-indexed framework for such structured families.

The main point is to replace the full ambient space $\mathbb{F}_q^k$ by an index set $\Omega$, together with a column map $v:\Omega\to\mathbb{F}_q^k$. If collisions between distinct messages are sufficiently sparse over $\Omega$, then the same puncture-and-append argument produces a code with large mutual correlated agreement error while keeping every coordinate inside the indexed family. We do not require the map $v$ to be injective, i.e., distinct indices may determine the same column. This is useful for evaluation codes, where the natural structural objects are evaluation points or places rather than distinct evaluation columns.

\subsection{A General Structure-Preserving Principle}\label{subsec:structured_principle}

We first formalize the indexed code family induced by a finite index set and a column map.

\begin{definition}\label{def:indexed_structured_family}
Let $\Omega$ be a finite set, and let $v:\Omega\to\mathbb{F}_q^k$ be a column map. The \emph{indexed structured linear-code family} induced by $(\Omega,v)$ consists of all $k$-dimensional linear codes $C\subseteq\mathbb{F}_q^n$ that have a full-row-rank generator matrix
\[
    G=\bigl(v(\omega_1)\mid v(\omega_2)\mid \cdots \mid v(\omega_n)\bigr)
    \in\mathbb{F}_q^{k\times n}
\]
for distinct indices $\omega_1,\ldots,\omega_n\in\Omega$.
\end{definition}

The key additional hypothesis for the puncture-and-append argument is a collision bound over the index set. The next lemma is the structured analogue of Lemma~\ref{lem:linear_collision_aware_out_of_domain}.

\begin{lemma}\label{lem:structured_collision_aware_indexed}
Let $\Omega$ be a finite set with $|\Omega|=N>0$, and let $v:\Omega\to\mathbb{F}_q^k$ be a map. Let $m_1,\ldots,m_{L+1}\in\mathbb{F}_q^k$ be distinct vectors. Assume that there is a number $\lambda$ such that, for every $1\leq i<j\leq L+1$,
\[
    \left|
        \{\omega\in\Omega:
        (m_i-m_j)\cdot v(\omega)^{\mathrm T}=0\}
    \right|
    \leq \lambda .
\]
Then there exists an index $\omega\in\Omega$ such that
\[
    \left|\{m_i\cdot v(\omega)^{\mathrm T}:1\leq i\leq L+1\}\right|
    \geq
    \left\lceil
        \frac{(L+1)N}{N+L\lambda}
    \right\rceil .
\]
\end{lemma}

\begin{proof}
For $\omega\in\Omega$, let
\[
    C(\omega)\triangleq
    \left|\{(i,j):1\leq i<j\leq L+1,\
    m_i\cdot v(\omega)^{\mathrm T}=m_j\cdot v(\omega)^{\mathrm T}\}\right|
\]
be the number of colliding pairs at the index $\omega$. By the assumed collision bound,
\[
    \sum_{\omega\in\Omega}C(\omega)
    =
    \sum_{1\leq i<j\leq L+1}
    \left|
        \{\omega\in\Omega:
        (m_i-m_j)\cdot v(\omega)^{\mathrm T}=0\}
    \right|
    \leq
    \binom{L+1}{2}\lambda .
\]
By averaging over all $N$ indices and applying the pigeonhole principle, there exists $\omega\in\Omega$ such that
\[
    C(\omega)\leq \frac{L(L+1)\lambda}{2N}.
\]
Fix such an $\omega$, and suppose that the values $m_1\cdot v(\omega)^{\mathrm T},\ldots,m_{L+1}\cdot v(\omega)^{\mathrm T}$ take exactly $r$ distinct values, with multiplicities $b_1,\ldots,b_r$. Then
\[
    \sum_{t=1}^r b_t=L+1,
    \qquad
    \sum_{t=1}^r b_t^2=(L+1)+2C(\omega).
\]
By the Cauchy--Schwarz inequality,
\[
    (L+1)^2
    \leq
    r\bigl((L+1)+2C(\omega)\bigr).
\]
Therefore
\[
    r
    \geq
    \frac{(L+1)^2}{(L+1)+2C(\omega)}
    \geq
    \frac{(L+1)^2}{(L+1)+L(L+1)\lambda/N}
    =
    \frac{(L+1)N}{N+L\lambda}.
\]
Since $r$ is an integer, the claimed bound follows.
\end{proof}

We can now repeat the puncture-and-append construction, replacing the full choice space $\mathbb{F}_q^k$ by the index set $\Omega$.

\ThmStructuredMCA*

\begin{proof}
By assumption, $C$ has a full-row-rank generator matrix
\[
    G=\bigl(v(\omega_1)\mid v(\omega_2)\mid\cdots\mid v(\omega_n)\bigr)
    \in\mathbb{F}_q^{k\times n}
\]
for distinct indices $\omega_1,\ldots,\omega_n\in\Omega$. Since $p\geq 0$ and $p<\Delta(C)-1/n$, we have $d_{\min}(C)>1$. Hence puncturing any single coordinate preserves dimension.

Let $y\in\mathbb{F}_q^n$ and distinct codewords $c_1,\ldots,c_{L+1}\in C$ witness that $C$ is not $(p,L)$-list-decodable. Write $c_i=m_iG$ with $m_i\in\mathbb{F}_q^k$. Since $G$ has full row rank, the messages $m_1,\ldots,m_{L+1}$ are distinct.

For every $i<j$, the difference $m_i-m_j$ is nonzero, and hence the assumed uniform collision bound gives
\[
    \left|
        \{\omega\in\Omega:
        (m_i-m_j)\cdot v(\omega)^{\mathrm T}=0\}
    \right|
    \leq \lambda .
\]
Thus we may apply Lemma~\ref{lem:structured_collision_aware_indexed} to the indexed family $(\Omega,v)$. We obtain an index $\omega_*\in\Omega$ such that
\(
    T\triangleq
    \{m_i\cdot v(\omega_*)^{\mathrm T}:1\leq i\leq L+1\}
\)
satisfies
\[
    |T|
    \geq
    \left\lceil
        \frac{(L+1)N}{N+L\lambda}
    \right\rceil .
\]

Choose the punctured coordinate $j^*\in[n]$ as follows. If $\omega_*=\omega_j$ for some $j\in[n]$, set $j^*=j$; otherwise choose any $j^*\in[n]$. Let $\widetilde{G}$ be the matrix obtained from $G$ by deleting the $j^*$-th column, let $\widetilde{C}$ be the corresponding punctured code, and let $\widetilde{y}$ and $\widetilde{c}_i=m_i\widetilde{G}$ be obtained from $y$ and $c_i$ by deleting the same coordinate. The indices remaining after puncturing are distinct and do not include $\omega_*$.

The punctured words still witness that $\widetilde{C}$ is not $(p',L)$-list-decodable for
\(
    p'\triangleq p\frac{n}{n-1},
\)
because puncturing cannot increase absolute Hamming distance. As in the proof of Lemma~\ref{lem:puncture_and_append}, puncturing can reduce the absolute minimum distance by at most one, so
\(
    \Delta(\widetilde{C})
    \geq
    \frac{n\Delta(C)-1}{n-1}.
\)
Thus
\(
    p'<\Delta(\widetilde{C}).
\)

Define
\[
    G'=(\widetilde{G}\mid v(\omega_*))
    \qquad\text{and}\qquad
    C'=\{mG':m\in\mathbb{F}_q^k\}.
\]
The matrix $G'$ has full row rank because $\widetilde{G}$ has full row rank, and it is indexed by the distinct indices obtained from $\omega_1,\ldots,\omega_n$ by deleting $\omega_{j^*}$ and appending $\omega_*$. In particular,
\[
    \left|
        \{\omega_1,\ldots,\omega_n\}
        \cap
        \{\omega'_1,\ldots,\omega'_n\}
    \right|
    \geq n-1 .
\]

Now define
\[
    \pi_0=(\widetilde{y},0),
    \qquad
    \pi_1=(0,\ldots,0,1)
    \in\mathbb{F}_q^n .
\]
The same argument as in Lemma~\ref{lem:out-of-domain} shows that every $\alpha\in T$ is a bad combining point for $C'$ at radius $p'\frac{n-1}{n}=p$. Indeed, if $\alpha=m_i\cdot v(\omega_*)^{\mathrm T}$, then $\pi_0+\alpha\pi_1$ agrees with the codeword
\[
    (\widetilde{c}_i,m_i\cdot v(\omega_*)^{\mathrm T})\in C'
\]
on the agreement set between $\widetilde{y}$ and $\widetilde{c}_i$, together with the appended coordinate. This set has size at least
\[
    (1-p')(n-1)+1
    =
    (1-p)n.
\]
On the other hand, $\pi_0$ and $\pi_1$ cannot have correlated agreement on this enlarged set. Indeed, such an agreement would give a codeword $(m\widetilde{G},m\cdot v(\omega_*)^{\mathrm T})\in C'$ that agrees with $\pi_1$ on the enlarged set. Since the appended coordinate of $\pi_1$ is $1$, we have $m\cdot v(\omega_*)^{\mathrm T}=1$, so $m\neq 0$. Also, $m\widetilde{G}$ vanishes on the agreement set between $\widetilde{y}$ and $\widetilde{c}_i$, and therefore
\[
    \operatorname{wt}(m\widetilde{G})
    \leq p'(n-1)
    <
    (n-1)\Delta(\widetilde{C}),
\]
contradicting the definition of $\Delta(\widetilde{C})$. Therefore the full agreement set between $\pi_0+\alpha\pi_1$ and the above codeword cannot belong to $\mathcal{A}_{p,\{\pi_0,\pi_1\},C'}$: if it did, then $\pi_0$ and $\pi_1$ would have correlated agreement on this full agreement set, and hence also on the enlarged subset considered above. Thus every $\alpha\in T$ is a bad combining point. Therefore
\[
    \left|\bad_{C',p}(\pi_0,\pi_1)\right|
    \geq |T|
    \geq
    \left\lceil
        \frac{(L+1)N}{N+L\lambda}
    \right\rceil .
\]
Dividing by $q$ and using Definition~\ref{def:mca_error} gives the claimed lower bound on $\mcaerr(C',p)$.

Finally, appending one coordinate cannot decrease the absolute weight of a nonzero codeword, while the puncturing step loses at most one in absolute minimum distance. Hence
\[
    d_{\min}(C')\geq n\Delta(C)-1,
    \qquad\text{so}\qquad
    \Delta(C')\geq \Delta(C)-\frac{1}{n}.
\]
\end{proof}

\subsection{Applications to Evaluation Codes}\label{subsec:structured_applications}

We now instantiate the general principle for two standard evaluation-code families. In both cases, the admissible columns come from evaluation: for AG codes they are indexed by rational places, while for Reed--Solomon codes they are Vandermonde columns indexed by elements of $\mathbb{F}_q$. The collision parameter is controlled by the number of zeros of a nonzero function or polynomial, respectively.

We first recall the small amount of notation needed for AG evaluation codes. For a divisor $G$ on an algebraic function field $F/\mathbb{F}_q$, let $\mathcal{L}(G)$ be the associated Riemann--Roch space, and write $\ell(G)\triangleq\dim_{\mathbb{F}_q}\mathcal{L}(G)$. If $D=P_1+\cdots+P_n$ is a sum of rational places disjoint from $\operatorname{supp}(G)$, then the AG evaluation code $C_{\mathcal{L}}(D,G)$ is obtained by evaluating functions in $\mathcal{L}(G)$ at $P_1,\ldots,P_n$. When $\deg G<n$, this evaluation map is injective, so the code has dimension $\ell(G)$.

\CorAGMCA*

\begin{proof}
Let $k=\ell(G)$ and fix a basis $h_1,\ldots,h_k$ of $\mathcal{L}(G)$. Each rational place $P\in\mathcal{P}$ determines an admissible evaluation column $v(P)\triangleq (h_1(P),\ldots,h_k(P))^{\mathrm T}\in\mathbb{F}_q^k$, and the generator matrix of $C_{\mathcal{L}}(D,G)$ has the form
\[
    \bigl(v(P_1)\mid\cdots\mid v(P_n)\bigr),
\]
where $P_1,\ldots,P_n$ are distinct elements of $\mathcal{P}$.

For every nonzero $u=(u_1,\ldots,u_k)\in\mathbb{F}_q^k$, the function $f_u\triangleq \sum_{r=1}^k u_rh_r$ is a nonzero element of $\mathcal{L}(G)$. Hence $f_u$ has at most $\deg G$ zeros among the rational places in $\mathcal{P}$, or equivalently
\[
    \left|\{P\in\mathcal{P}:u\cdot v(P)^{\mathrm T}=0\}\right|\leq \deg G .
\]
Apply Theorem~\ref{thm:structured_puncture_append} with $\Omega=\mathcal{P}$ and $\lambda=\deg G$. It produces distinct places $P'_1,\ldots,P'_n\in\mathcal{P}$, a corresponding AG evaluation code $C_{\mathcal{L}}(D',G)$ for $D'=P'_1+\cdots+P'_n$, and the claimed mutual correlated agreement lower bound and distance bound. Since the theorem constructs the new index set by puncturing one original place and appending one admissible place, we may write
\[
    D'=D-P_j+P
    \qquad\text{with}\qquad
    P\in\mathcal{P}\setminus\operatorname{supp}(D-P_j).
\]
\end{proof}

The Reed--Solomon case is the specialization in which the admissible columns are Vandermonde columns.

\begin{corollary}[Reed--Solomon codes]\label{cor:rs_structured_framework}
Let $S\subseteq\mathbb{F}_q$ be a set of $n$ distinct evaluation points. Suppose that $n>k$, $\rs_{\mathbb{F}_q}(S,k)$ is not $(p,L)$-list-decodable, and
\(
    p<1-\frac{k}{n}.
\)
Then there exists a set $S'\subseteq\mathbb{F}_q$ of $n$ distinct evaluation points such that
\[
    \mcaerr\bigl(\rs_{\mathbb{F}_q}(S',k),p\bigr)
    \geq
    \frac{1}{q}
    \left\lceil
        \frac{(L+1)q}{q+L(k-1)}
    \right\rceil .
\]
Moreover, $\rs_{\mathbb{F}_q}(S',k)$ has the same length, dimension, and minimum distance as $\rs_{\mathbb{F}_q}(S,k)$.
\end{corollary}

\begin{proof}
For each $x\in\mathbb{F}_q$, let
\[
    v(x)\triangleq (1,x,\ldots,x^{k-1})^{\mathrm T}\in\mathbb{F}_q^k
\]
be the Vandermonde column at $x$. Apply the indexed framework with $\Omega=\mathbb{F}_q$, so $N=q$. For every nonzero $u=(u_0,\ldots,u_{k-1})\in\mathbb{F}_q^k$, the condition $u\cdot v(x)^{\mathrm T}=0$ is the equation
\[
    u_0+u_1x+\cdots+u_{k-1}x^{k-1}=0,
\]
which has at most $k-1$ solutions in $\mathbb{F}_q$. Thus the collision parameters in Theorem~\ref{thm:structured_puncture_append} are $N=q$ and $\lambda=k-1$. Since
\[
    \Delta\bigl(\rs_{\mathbb{F}_q}(S,k)\bigr)-\frac{1}{n}
    =
    1-\frac{k}{n},
\]
the distance hypothesis is exactly the one required by the theorem.

Theorem~\ref{thm:structured_puncture_append} gives a set $S'\subseteq\mathbb{F}_q$ of $n$ distinct evaluation points and the stated mutual correlated agreement lower bound. Since both the original and resulting codes are Reed--Solomon codes of length $n$ and dimension $k$, both have minimum distance $n-k+1$.
\end{proof}

\printbibliography

\end{document}